\documentclass[conference]{IEEEtran}

\pdfoutput=1
\usepackage{graphicx}
\usepackage{amssymb}
\usepackage{epstopdf}
\usepackage{float}
\usepackage{dblfloatfix}
\usepackage{amsmath}
\usepackage{algorithm}
\usepackage[noend]{algpseudocode}
\usepackage[justification=centering]{caption}
\usepackage{amsthm}
\theoremstyle{definition}
\newtheorem{definition}{Definition}
\newtheorem{example}{Example}
\newtheorem{remark}{Remark}

\newtheorem{lemma}{Lemma}
\newtheorem{theorem}{Theorem}

\begin{document}
\title{Noisy Index Coding with Quadrature Amplitude Modulation (QAM)}
\author{Anjana~A.~Mahesh  and B~Sundar~Rajan,~\IEEEmembership{Fellow,~IEEE,}}
\maketitle
\begin{abstract}
This paper discusses noisy index coding problem over Gaussian broadcast channel. We propose a technique for mapping the index coded bits to M-QAM symbols such that the receivers whose side information satisfies certain conditions get coding gain, which we call the \textbf{QAM side information coding gain}. We compare this with the PSK side information coding gain, which was discussed in \cite{ICPSKM}.\footnote{The authors are with the Department of Electrical Communication Engineering, Indian Institute of Science, Bangalore-560012, India, Email: anjanaam@ece.iisc.ernet.in, bsrajan@ece.iisc.ernet.in}
\end{abstract}

\section{Introduction and Preliminaries} 
\label{sec1}
The problem of index coding over noiseless broadcast channels was introduced in \cite{ISCO} and has been well studied \cite{ICSI} - \cite{OMIC}. It involves a single source and a set of caching receivers. Each of the receivers wants a subset of the set of messages transmitted by the source and knows another non-intersecting subset of messages a priori as side information. The problem is to minimize the number of binary transmissions required to satisfy the demands of all the receivers, which amounts to minimizing the bandwidth required. 

An index coding problem $\lbrace \mathcal{X}, \mathcal{R} \rbrace$, involves a single source, $\mathcal{S}$ that wishes to send a set of $n$ messages $\mathcal{X}=\lbrace x_1, x_2, \ldots, x_n \rbrace$ to a set of $m$ receivers, $\mathcal{R} = \lbrace R_1, R_2, \ldots, R_m \rbrace$. The messages, $x_i$, $i\in \lbrace1,2, \ldots, n \rbrace$ take values from some finite field $\mathbb{F}$. A receiver $R_i$, $i\in \lbrace1,2, \ldots, m \rbrace$, is defined as $R_i = \lbrace \mathcal{W}_i, \mathcal{K}_i \rbrace$. $\mathcal{W}_i \subseteq \mathcal{X}$ is the set of messages demanded by  $R_i$ and $\mathcal{K}_i \subsetneq \mathcal{X}$ is the set of messages known to $R_i$, a priori, known as the side information that $R_i$ has.

An index code for the index coding problem with $\mathbb{F}=\mathbb{F}_2$ consists of
\begin{enumerate}
\item  an encoding map, $f:\mathbb{F}_{2}^n \rightarrow \mathbb{F}_{2}^l$, where $l$ is called the length of the index code, and
\item  a set of decoding functions $g_{1}, g_{2},\ldots,g_{m}$ such that, for a given input $\textbf{x} \in \mathbb{F}_{2}^n $, $g_{i}\left(f(\textbf{x}),\mathcal{X}_i\right) = \mathcal{W}_i, \ \forall i \in \lbrace 1, 2,\ldots,m \rbrace $.
\end{enumerate}
An optimal index code for binary transmissions  minimizes $l$, the number of binary transmissions required to satisfy the demands of all receivers.

A linear index code is one whose encoding function is linear and it is linearly decodable if all the decoding functions are linear. It was shown in \cite{ICSI} that for the class of index coding problems over $\mathbb{F}_{2}$ which can be represented using side information graphs, which were labeled later in \cite{OMIC} as single unicast index coding problems, the length of optimal linear index code is equal to the minrank over $\mathbb{F}_{2}$ of the corresponding side information graph. This was extended in \cite{ECIC} to  general index coding problems, over $\mathbb{F}_q$,  using minrank over $\mathbb{F}_q$ of their corresponding side information hypergraphs.

In this paper, we consider noisy index coding problems with $\mathbb{F}=\mathbb{F}_2$ over AWGN channels. In the noisy version of index coding, the messages are sent by the source over a noisy broadcast channel. Instead of binary transmissions if multilevel ($M$-ary, $M > 2$) modulation schemes are used further bandwidth reduction can be  achieved.  This has been introduced in \cite{IGBC} for Gaussian broadcast channels. It was also found in \cite{IGBC} that using $M$-ary modulation has the added advantage of giving coding gain to receivers with side information, termed the "side information gain". The idea of side information gain was characterized for the case where the source use $M$-PSK for transmitting the index coded bits in \cite{ICPSKM}, where  $M=2^l$, where $l$ is the length of the index code used, with the average energy of the $M$-QAM signal being equal to the total energy of $l$ binary transmissions. 

This paper discusses the case of noisy index coding over AWGN channels where the source uses $M$-QAM to transmit the index coded bits. As in \cite{ICPSKM}, here also, M = $2^l$, where $l$ is the length of the index code used. 

The contributions and organization in this paper may be summarized as follows:
\begin{enumerate}
\item An algorithm to map binary symbols to appropriate sized QAM constellation is presented which uses the well known Ungerboeck labelling as an ingredient. (Section \ref{sec3}
\item A necessary and sufficient condition for a receiver to get side information coding gain is presented. (Theorem \ref{thm1} in Section \ref{sec4})
\item It is shown that the difference in probability of error performance between the best and worst performing receivers increases monotonically as the length of the index code used increases. (Theorem \ref{thm2} in Section \ref{sec4})
\end{enumerate}
\noindent
In Section \ref{sec2} the notions of bandwidth gain and QAM side information coding gain are explained and simulation results are presented in Section \ref{sec5}. Concluding remarks constitute Section \ref{sec6} 

	
\section{Bandwidth Gain and QAM Side Information Coding Gain}
\label{sec2}
	
Consider a general index coding problem $\lbrace \mathcal{X}, \mathcal{R} \rbrace$ with $n$ messages, $\mathcal{X}=\lbrace x_1, x_2, \ldots, x_n \rbrace$ and $m$ receivers, $\mathcal{R} = \lbrace R_1, R_2, \ldots, R_m \rbrace$. Let the length of a linear index code (not necessarily optimal) for the index coding problem at hand be $l$.  We have $N \leq l \leq n$, where $N$ is the length of the optimal index code which is equal to the minrank over $\mathbb{F}_2$ of the corresponding side information hypergraph. Let the encoding matrix corresponding to the linear index code chosen be $L$, where $L$ is an $n \times l$ matrix over $\mathbb{F}_2$. The index coded bits are given by $\textbf{y} = \left[y_1 \  y_2 \dots y_l\right] = \textbf{x}L,$ where $\textbf{x} = \left[x_1 \  x_2 \dots x_n\right].$

The noiseless index coding involves $l$ binary transmissions. It was shown in \cite{ICPSKM} that for noisy index coding problems, if we transmit the index coded bits as a point from $2^l$-PSK signal set instead of $l$ binary transmissions, the receivers  satisfying certain conditions will get coding gain in addition to bandwidth gain whereas other receivers trade off coding gain for bandwidth gain. This gain which was termed as the "PSK side information coding gain" was obtained by proper mapping of index coded bits to PSK symbols an algorithm for which was presented. In this paper, we extend the results in \cite{ICPSKM} for the case where we use $2^l$-QAM to transmit the $l$ index coded bits.
	
\begin{definition}
The term \textbf{QAM bandwidth gain} is defined as the bandwidth gain obtained by each receiver by going from $l$ binary transmissions to a single $2^l$-QAM symbol transmission.		
\end{definition}

When we transmit a single $2^l-QAM$ signal point instead of transmitting $l$ binary transmissions we are going from an $l$- real dimensional or equivalently $l/2$ - complex dimensional signal set to 1 complex dimensional signal set. Hence all the receivers get a $l/2$ - fold QAM bandwidth gain. We state this simple fact as 
	
\begin{lemma}
\label{Lem:QAM BG}
Each receiver gets an $l/2$ - fold QAM bandwidth gain.
\end{lemma}

\begin{definition}
The term \textbf{QAM side information coding gain} (QAM-SICG) is defined as the coding gain a receiver with a non-empty side information set gets w.r.t a receiver with no side information while using $2^l$-QAM to transmit the index coded bits. 
\end{definition}

Let the set $S_{i}$, $i \in \lbrace1,2, \ldots, m \rbrace$, be defined as the set of all binary transmissions which a receiver $R_{i}$ knows a priori due to its available side information, i.e.,
$$S_{i}= \lbrace y_{j}|y_{j}=\sum\limits_{k \in J }x_{k} ,\ J\subseteq \mathcal{K}_{i}\rbrace .$$
Also, let $\eta_i$, $ i \in \lbrace1,2, \dots, m \rbrace $ be defined as follows.
   $$\eta_i \triangleq min \lbrace n - \left|\mathcal{K}_i\right|, \ l - \left|S_i\right| \rbrace.$$

\section{Algorithm}
\label{sec3}
In this section, we describe an algorithm to map the index coded bits to signal points of an appropriate sized QAM constellation so that the receivers satisfying the conditions of Theorem \ref{thm1} in the following section will get QAM-SICG. For the given index coding problem, choose an index code of length $l$. This fixes the value of $\eta_i, \ i \in \lbrace 1,2, \dots, m \rbrace$.
	
Order the receivers in the non-decreasing order of $\eta_{i}$.
WLOG, let $\lbrace R_{1},R_{2},.. ,R_{m} \rbrace$ be such that $\eta_{1} \leq \eta_2 \leq \ldots \leq \eta_m$.
	
Before starting to run the algorithm to map the index coded bits to $2^l$-QAM symbols, we need to

\begin{enumerate}
\item{Choose an appropriate $2^l$-QAM signal set.}
\item{Use Ungerboeck set partitioning \cite{TCM} to partition the $2^l$-QAM signal set chosen into subsets with increasing minimum subset distances}.
\end{enumerate}
	
To choose the appropriate QAM signal set, do the following: 
\begin{itemize}
\item \textbf{if} $l$ is even, then choose the $2^l$-square QAM with average symbol energy being equal to $l$.
\item \textbf{else}, take the $2^{l+1}$-square QAM with average symbol energy equal to $l$. Use Ungerboeck set partitioning \cite{TCM} to partition the $2^{l+1}$ QAM signal set into two $2^l$ signal sets. Choose any one of them as the $2^l$-QAM signal set.
\end{itemize}
	
Let $L_0,\ L_1, ..., L_{l-1}$ denote the different levels of partitions of the $2^l$-QAM with the minimum distance at layer $L_i=\Delta_i$, $i \in \lbrace0, 1,\ldots, l-1\rbrace$, being such that $\Delta_0 < \Delta_1 <  \ldots < \Delta_{l-1}$.
	
The algorithm to map the index coded bits to QAM symbols is given in \textbf{Algorithm 1}.

\begin{algorithm}
\caption{Algorithm to map index coded bits to QAM symbols}\label{algo1}
\begin{algorithmic}[1]
		
\If {$\eta_1 \geq N $}, do an arbitrary order mapping and \textbf{exit}.
\EndIf

			\State $i \gets 1$
			
			\If {all $2^N$ codewords have been mapped}, \textbf{exit}.
			\EndIf
			
			\State  Fix $( x_{i_1},x_{i_2}, \ldots, x_{i_{\left|\mathcal{K}_i\right|}} )=(a_{1},a_{2}, \ldots, a_{\left|\mathcal{K}_i\right|}) \in \mathcal{A}_i$ such that the set of codewords, $\mathcal{C}_i \subset \mathcal{C} $, obtained by running all possible combinations of $\lbrace x_{j}|\ j \notin \mathcal{K}_i\rbrace$ with $( x_{i_1},x_{i_2}, \ldots, x_{i_{\left|\mathcal{K}_i\right|}} )=(a_{1},a_{2}, \ldots, a_{\left|\mathcal{K}_i\right|})$ has maximum overlap with the codewords already mapped to PSK signal points.
			
			\If {all codewords in $\mathcal{C}_i$ have been mapped},
			\begin{itemize}
				\item $\mathcal{A}_i$=$\mathcal{A}_i \setminus \lbrace( x_{i_1},x_{i_2}, \ldots, x_{i_{\left|\mathcal{K}_i\right|}} )|( x_{i_1},x_{i_2}, \ldots, x_{i_{\left|\mathcal{K}_i\right|}} )$ together with all combinations of $\lbrace x_{j}|\ j \notin \mathcal{K}_i\rbrace$ will result in $\mathcal{C}_i\rbrace$.
				\item $i \gets i+1$
				\item \textbf{if} {$\eta_i \geq N$} \textbf{then},
				\begin{itemize}
					\item $i \gets 1$.
					\item goto \textbf{Step 3}
				\end{itemize} 
				\item \textbf{else}, goto \textbf{Step 3}
				
			\end{itemize}

			\Else
			\begin{itemize}
				\item Of the codewords in $\mathcal{C}_i$ which are yet to be mapped, pick any one and map it to a QAM signal point in that $2^{\eta_i}$ sized subset at level $L_{l-\eta_i}$ which has maximum number of signal points mapped by codewords in $\mathcal{C}_i$ without changing the already labeled signal points in that subset. 
				
				If all the signal points in such a subset have been already labeled, then map it to a signal point in another $2^{\eta_i}$ sized subset at the same level $L_{l-\eta_i}$  that this point together with the signal points corresponding to already mapped codewords in $\mathcal{C}_i$, has the largest minimum distance possible. Clearly this minimum distance, $d_{min}(R_i)$ is such that $\Delta_{l-\eta_i} \geq d_{min}(R_i) \geq \Delta_{l-(\eta_i+1)}$.
				\item $i \gets 1$
				\item goto \textbf{Step 3}
			\end{itemize}
			\EndIf

		\end{algorithmic}
	\end{algorithm}
	
		\begin{remark}
			Note that Algorithm \ref{algo1} above does not result in a unique mapping of index coded bits to $2^l$-QAM symbols. The mapping will change depending on the choice of $( x_{i_1},x_{i_2}, \ldots, x_{i_{\left|\mathcal{K}_i\right|}} )$ in each step. However, the performance of all the receivers obtained using any such mapping scheme resulting from the algorithm will be the same. 
		\end{remark}
		
		\begin{remark}
			If $\eta_{i}= \eta_j$ for some $i \neq j$, depending on the ordering of $\eta_i$ done before starting the algorithm, $R_i$ and $R_j$ may give different performances in terms of probability of error. $R_i$ and $R_j$ with $\eta_{i}= \eta_j$ will give the same performance if and only if $S_i \subseteq S_j$ or vice-versa.
		\end{remark}

\section{Main results}
\label{sec4}
In this section we present the main results apart from the algorithm given in the previous section. 
\begin{theorem}
\label{thm1}
A receiver $R_i, \ i \in \lbrace 1,2, \ldots, m \rbrace$ gets QAM side information coding gain,  with the scheme proposed, if and only if $\eta_i < l$, where $l$ is the length of the index code used.
\end{theorem}
\begin{proof} Consider a receiver $R_i = \lbrace \mathcal{W}_i, \mathcal{K}_i \rbrace$.      
 Let $\mathcal{K}_i = \lbrace i_{1}, i_2, \ldots, i_{\left|\mathcal{K}_i\right|} \rbrace$ and $\mathcal{A}_i$ $\triangleq$ $\mathbb{F}_2^{\left|\mathcal{K}_i\right|}$, $i=1,2,\ldots,m$. For any given realization of $( x_{i_1},x_{i_2}, \ldots, x_{i_{\left|\mathcal{K}_i\right|}} )$, the effective signal set seen by the receiver $R_i$ consists of $2^{\eta_i}$ points. Let $d_{min}(R_i) \triangleq$ the minimum distance of the signal set seen by the receiver $R_i$, $i = 1,2,\ldots,m.$

\noindent
{\it Proof of the `if part':} If $\eta_i < l$, then the effective signal set seen by the receiver $R_i$ will have $2^{\eta_i} < 2^l$ points. Hence by appropriate mapping of index coded bits to QAM symbols, we can increase $d_{min}(R_i)$. Thus $R_i$ will get coding gain over a receiver that has no side information because the minimum distance seen by a receiver with no side information will be the minimum distance of $2^l$-QAM signal set. 
        
\noindent
{\it Proof of the `only if part':} Let us a consider a receiver $R_i$ such that $\eta_{i} \geq l $. Then $d_{min}(R_i)$, will not increase. $d_{min}(R_i)$ will remain equal to the minimum distance of the  corresponding $2^{l}$ - QAM, same as that of a receiver with no side information. Thus a receiver $R_i$ with $\eta_{i} \geq l $ will not get QAM-SICG.
\end{proof}
\begin{remark}
It is to be noted that the value of $\eta_i$ not only depends on $\left|\mathcal{K}_i\right|$ but also on $\left|S_i\right|$, which, in turn, depends on the index code chosen. Hence for the same index coding problem, a particular receiver may satisfy Theorem \ref{thm1} and get QAM-SICG for some index codes and may not get QAM-SICG for other index codes.
\end{remark}

\begin{theorem}
\label{thm2}
The difference in probability of error performance between the best performing receiver and the worst performing receiver for a given index coding problem, while using $2^l$- QAM signal point to transmit the index coded bits, will increase monotonically while $l$ increases from $N$ to $n$ if the following conditions are satisfied.\\
(1) The best performing receiver gets QAM-SICG.\\
(2) The worst performing receiver has no side information.
\begin{proof}
If there is a receiver with no side information, say $R$, whatever the length, $l$, of the index code used is, the effective signal set seen by $R$ will be $2^l$-QAM. Therefore the minimum distance seen by $R$ will be the minimum distance of $2^l$-QAM signal set. For $2^l$-QAM with average symbol energy equal to $l$, the squared minimum pair-wise distance of $2^l$-QAM, $d_{min}(2^l$-QAM$)$, obtained by the proposed mapping scheme in Algorithm \ref{algo1} is given by 
                        
\begin{equation*}
d_{min}(2^l-\text{QAM})=\left\{
\begin{array}{@{}ll@{}}
2\sqrt{\dfrac{1.5l}{(2^{l}-1)}}, & \text{if}\ l \text{ is even} \\
2\sqrt{2}\sqrt{\dfrac{1.5l}{(2^{l+1}-1)}}, & \text{otherwise}
\end{array}\right.
\end{equation*}                 

\noindent                
which is monotonically decreasing in $l$. Therefore the performance of the receiver with no side information deteriorates as the length of the index code increases from $N$ to $n$. 
\end{proof}
\end{theorem}
\begin{remark}
Although condition (1) in Theorem \ref{thm2} is necessary, the same cannot be said about condition (2). Even when condition (2) above is not satisfied, i.e., the worst performing receiver has at least 1 bit of side information but not the same amount of side information as the best performing receiver, the difference between their performances can still increase monotonically as we move from $N$ to $n$. This is because the error performance is determined by the effective minimum distances seen by the receivers, which, in turn, depend on the mapping used between index coded bits and QAM symbols.
\end{remark}
\section{Simulation results}
\label{sec5}
A general index coding problem can be converted into one where each receiver demands only one message since a receiver $R_i = \lbrace \mathcal{W}_i, \mathcal{K}_i \rbrace$ can be converted into $\left|\mathcal{W}_i\right|$ receivers all with the same side information $\mathcal{K}_i$ and each demanding a single message. So it is enough to consider index coding problems where the receivers demand a single message each and hence both the examples considered in this section are such problems. Even though the examples considered are what are called single unicast index coding problems in \cite{OMIC}, the results hold for any general index coding problem.

\subsection{QAM-SICG and QAM Vs PSK}
\label{sec:Comparison}
In this subsection, we give an example with simulation results to support our claims in Section \ref{sec4}. The mapping of index coded bits to QAM symbols is done using our Algorithm \ref{algo1}. The receivers which satisfy Theorem \ref{thm1} are shown to get QAM-SICG. We also compare the performance of different receivers while using QAM and PSK to transmit index coded bits. For a given index coding problem and a chosen index code, the mapping of index coded bits to QAM symbols is done using the algorithm described in Section \ref{sec3}, whereas the mapping to PSK symbols is done using the Algorithm 1 in \cite{ICPSKM}. We also give the effective minimum distances which are seen by different receivers which explains the difference in their error performance.
\begin{example}
\label{ex_PSK_QAM} 
Let $m=n=$7. $\mathcal{W}_i = x_{i},\forall i\in \lbrace 1, 2,\ldots,7 \rbrace $. 
$\mathcal{K}_1 =\left\{2,3,4,5,6,7\right\},\ \mathcal{K}_2=\left\{1,3,4,5,7\right\},\ \mathcal{K}_3=\left\{1,4,6,7\right\},\mathcal{K}_4=\left\{2,5,6\right\},\mathcal{K}_5=\left\{1,2\right\},\ \mathcal{K}_6=\left\{3\right\},\mathcal{K}_7=\phi$.\\
The minrank over $\mathbb{F}_{2}$ of the side information graph corresponding to the above problem evaluates to $N$=4. 
An optimal linear index code is given by the encoding matrix,\\
\begin{center}
$L =\left[\begin{array}{cccc}
1 & 0 & 0 & 0\\
1 & 0 & 0 & 0\\
0 & 1 & 0 & 0\\
0 & 0 & 1 & 0\\
1 & 0 & 0 & 0\\
0 & 1 & 0 & 0\\
0 & 0 & 0 & 1
\end{array}\right]$. 
\end{center}
The index coded bits are, $y_{1}=x_{1}+x_{2}+x_{5}$; $y_{2}=x_{3}+x_{6}$; $y_{3}=x_{4}$;  $y_{4}=x_{7}$.
\end{example}
\noindent
The 16-QAM mapping for the above example is given in Fig. \ref{Map_ex-PSK_QAM}. The simulation result which compares the performance of different receivers when they use 16-QAM and 16-PSK for transmission of index coded bits is shown in Fig. \ref{sim_ex-PSK_QAM}. The probability of error plot corresponding to 4-fold binary transmission is also shown in Fig. \ref{sim_ex-PSK_QAM}. The reason for the difference in performance while using QAM and PSK can be explained using the minimum distance seen by the different receivers for the 2 cases. This is summarized in TABLE \ref{Table-ex_PSK_QAM}.
\begin{figure}
\includegraphics[scale= 0.6]{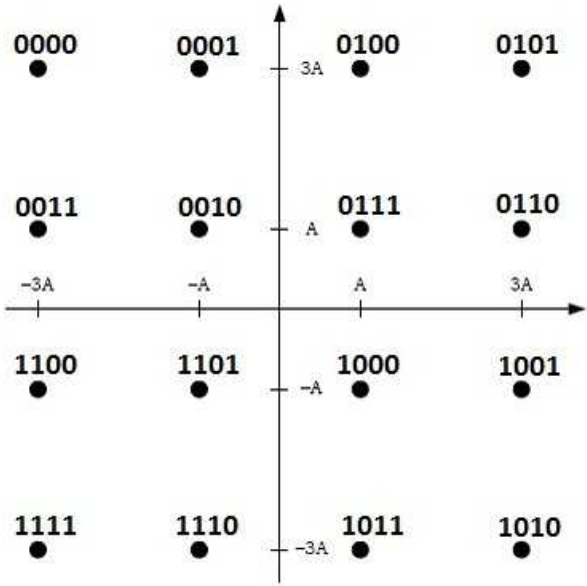}
\caption{16-QAM mapping for Example \ref{ex_PSK_QAM}}
\label{Map_ex-PSK_QAM}	
\end{figure}


\begin{table}[h]
		\renewcommand{\arraystretch}{2}
		\begin{center}
			
			\begin{tabular}{|c|c|c|c|c|c|c|c|}
				\hline
				Parameter & $R_{1}$ & $R_{2}$ & $R_{3}$ & $R_{4}$ & $R_{5}$ & $R_{6}$ & $R_{7}$ \\
				\hline 
				
			$d_{min}^2 - 16-QAM$ & 12.8 & 6.4 & 6.4 & 1.6 & 1.6 & 1.6 & 1.6 \\ 
				
			$d_{min}^2 - 16-PSK$ & 16 & 8 & 8 & 0.61 & 0.61 & 0.61 & 0.61  \\ 
					
			$d_{min}^2 - binary$ & 4 & 4 & 4 & 4 & 4 & 4 & 4 \\ 				
			\hline
				
			\end{tabular}
			
			\caption \small { Table showing  minimum distance seen by different receivers while using 16-QAM and 16-PSK in Example \ref{ex_PSK_QAM}.}
			
			\label{Table-ex_PSK_QAM}	
			
		\end{center}
\end{table}
	
\subsection {Performance for different QAM sizes-$2^N$ to $2^n$}
\label{sec:N_to_n}

In this subsection, we give an example to support our main results in Section \ref{sec3} that the difference in probability of error performance between the best performing receiver and the worst performing receiver widens as the length of the index code increases from $N$ to $n$. 

\begin{example}
\label{ex_N_n}
Let $m=n=5.\ \mathcal{W}_i = \lbrace x_i \rbrace, \ \forall i \in \lbrace 1, 2, 3, 4, 5 \rbrace.\ \mathcal{K}_1=\lbrace2,3,4,5\rbrace, \ \mathcal{K}_2=\lbrace1,3,5\rbrace, \ \mathcal{K}_3=\lbrace1,4\rbrace, \ \mathcal{K}_4= \lbrace2\rbrace, \ \mathcal{K}_5 = \phi$.
\noindent	
For this problem, minrank, $N$ = 3. An optimal linear index code is given by $L_1$ with the index coded bits being
$ y_{1}=x_{1}+x_{2}+x_{3};~~ y_{2}=x_{2}+x_{4}; ~~ y_{3}=x_{5}. 
$
\begin{align*}
L_1 = \left[\begin{array}{ccc}
1 & 0 & 0\\
1 & 1 & 0\\
1 & 0 & 0\\
0 & 1 & 0\\
0 & 0 & 1
\end{array}\right],\ L_2 = \left[\begin{array}{cccc}
1 & 0 & 0 & 0\\
1 & 0 & 0 & 0\\
0 & 1 & 0 & 0\\
0 & 0 & 1 & 0\\
0 & 0 & 0 & 1
\end{array}\right]
\end{align*}

%
%
Now, consider an index code of length $N+1=4$. The corresponding encoding matrix is $L_2$ 
and the index coded bits are
$ y_{1}=x_{1}+x_{2};~~ y_{2}=x_{3}; y_{3}=x_{4}; y_{4}=x_{5}.
$
We compare these with the case where we send the messages as they are, i.e., $L_3= I_5$, where $I_5$ denotes the $5 \times 5$ identity matrix.

The QAM mappings which give performance advantage to receivers satisfying conditions (1) and (2) of Theorem \ref{thm2} given in Section \ref{sec3} for the three different cases considered are given in Fig. \ref{Map_ex-N_n}(a) and (b) and \ref{Map_ex_n} respectively.
\begin{figure*}
	\includegraphics[scale= 0.7]{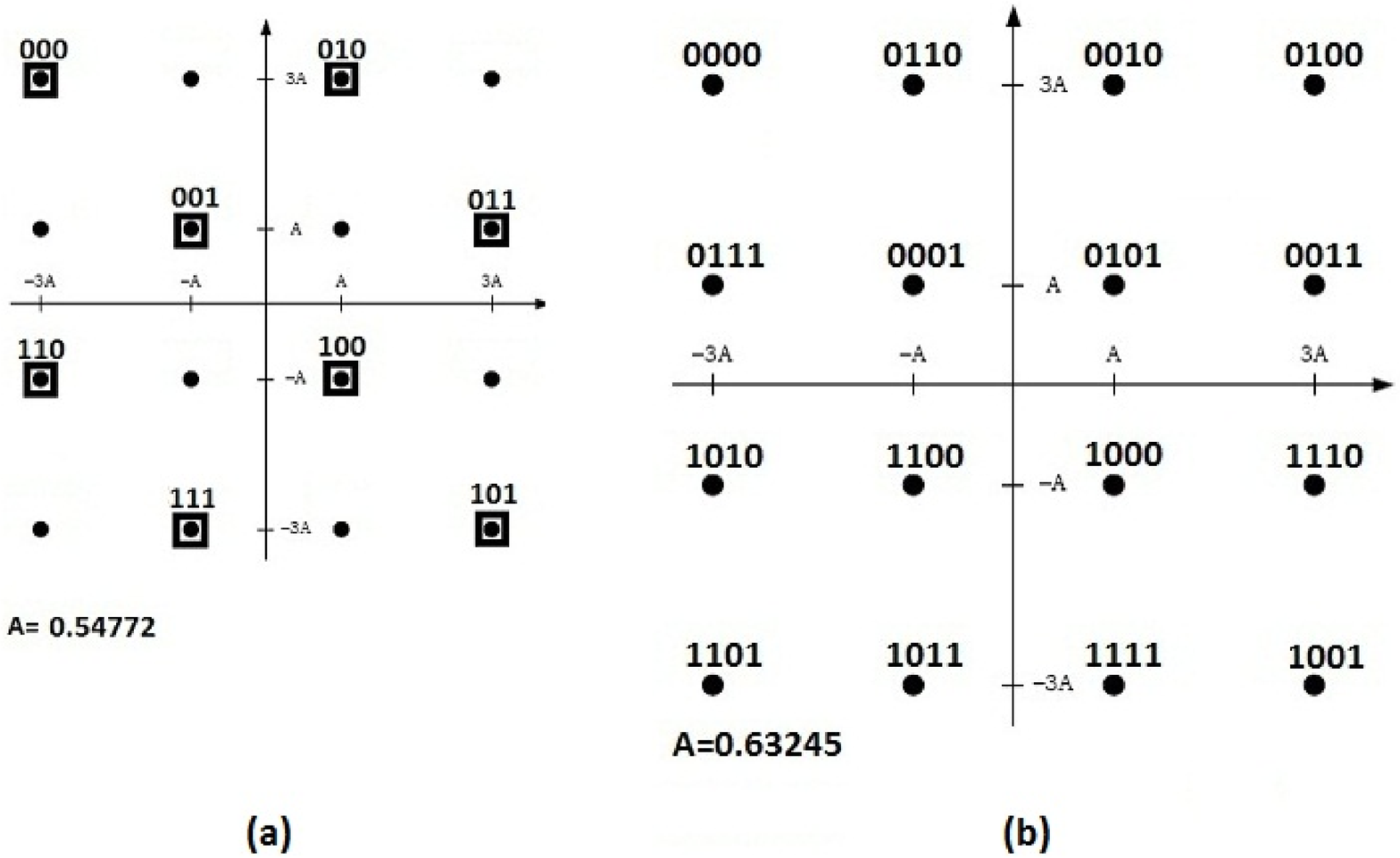}
	\caption{8-QAM and 16-QAM mapping for Example \ref{ex_N_n}}
	\label{Map_ex-N_n}	
\end{figure*}

\begin{figure}
	\includegraphics[scale= 0.5]{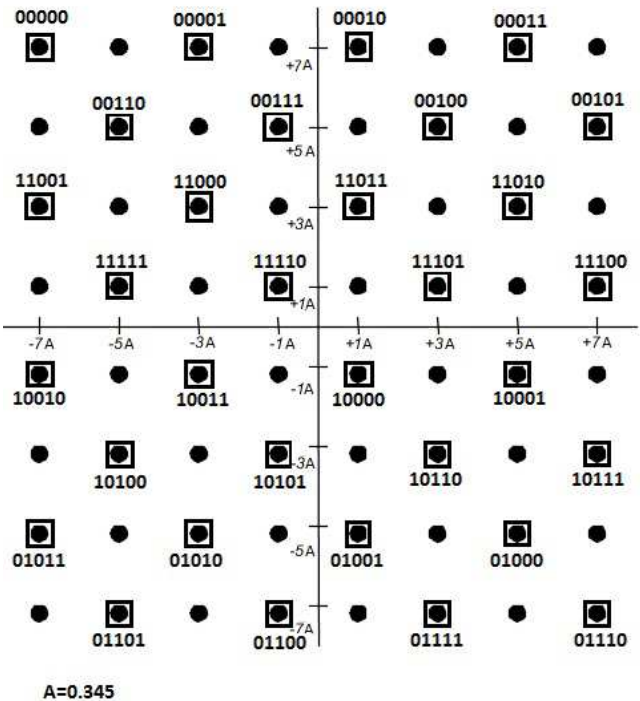}
	\caption{32-QAM mapping for Example \ref{ex_N_n}}
	\label{Map_ex_n}	
\end{figure}
\begin{figure*}[h]
	\includegraphics[scale=0.4]{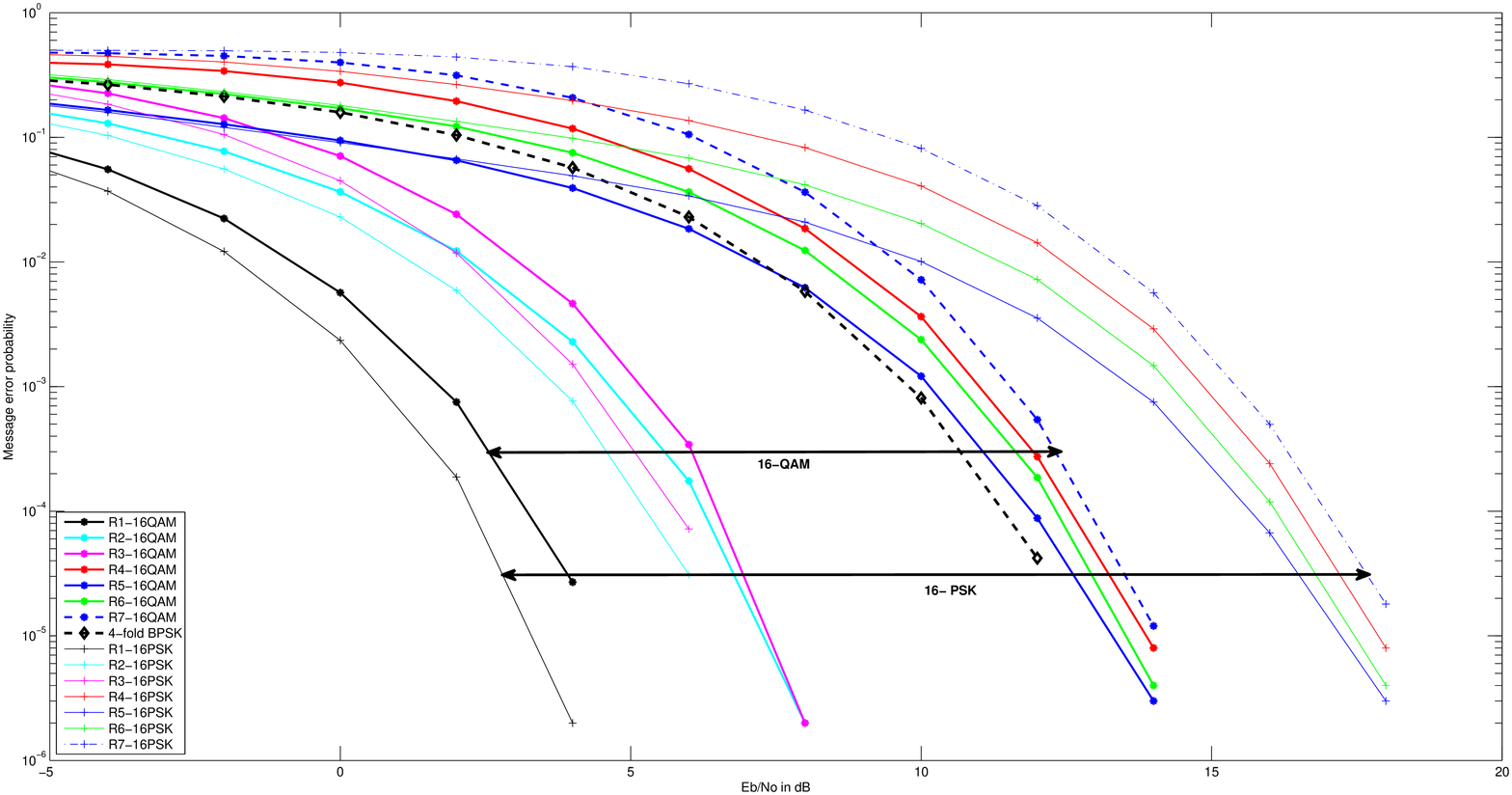}
	\caption{Simulation result comparing the performance of 16-PSK and 16-QAM for Example \ref{ex_PSK_QAM}.}
	\label{sim_ex-PSK_QAM}
\end{figure*}

\begin{figure*}[h]
	\includegraphics[scale=0.4]{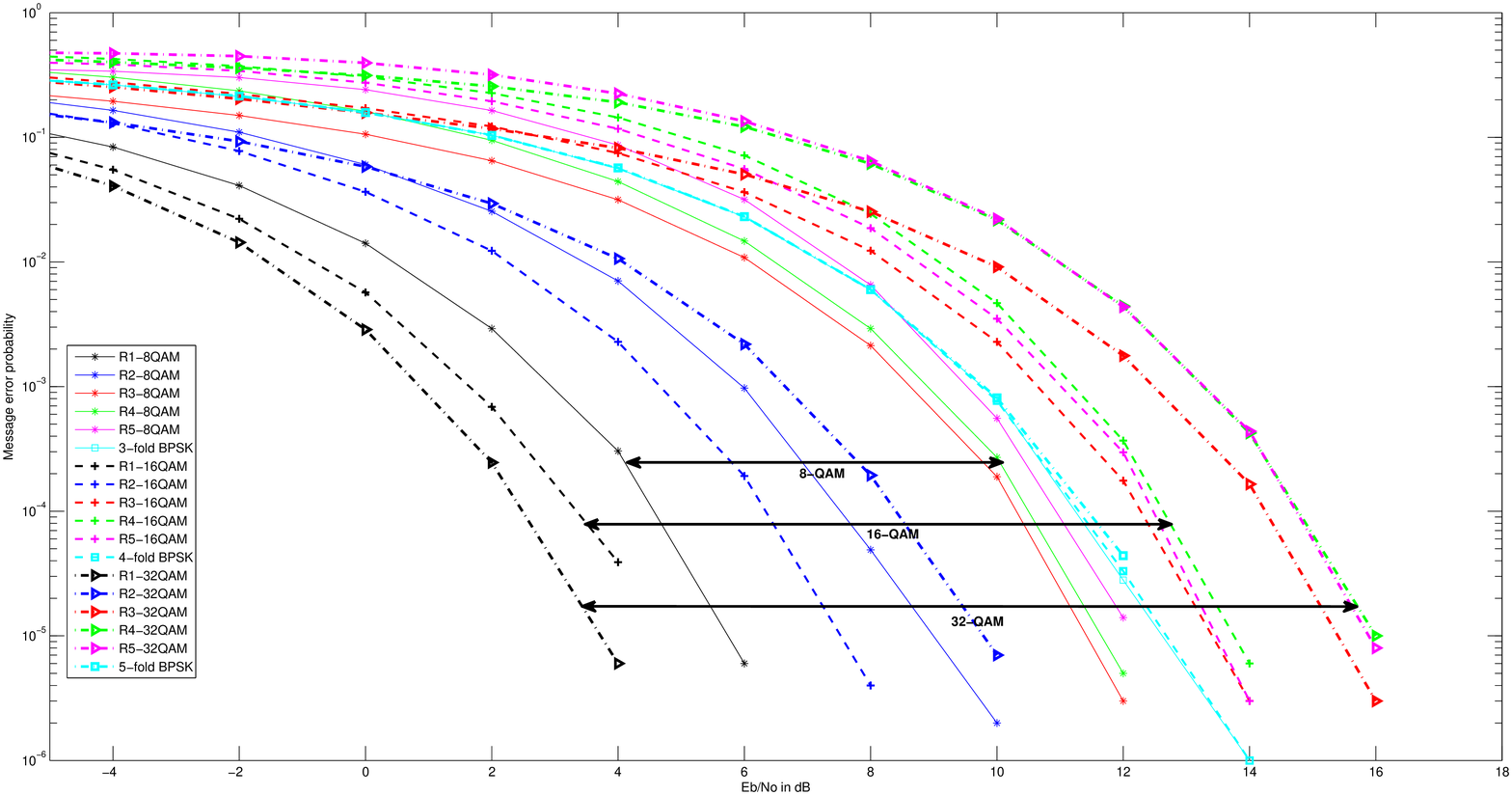}
	\caption{Simulation result comparing the performance of 8-QAM, 16-QAM and 32-QAM for Example \ref{ex_N_n}.}
	\label{sim_ex-N_n}
\end{figure*}
The simulation results for the three cases considered in this example are shown in Fig. \ref{sim_ex-N_n}. The difference in performance shown by the different receivers while using different sized QAM signal sets is because of the difference in the effective minimum distance seen by different receivers while using different signal sets corresponding to index codes of increasing lengths. The effective minimum distances seen by the receivers are summarized in the TABLE \ref{Table_ex_N_n} below. The difference in performance between receivers seeing the same minimum distance is because of the different distance distributions seen by them. 
\begin{table}[h]
	\renewcommand{\arraystretch}{2}
	\begin{center}
		
		\begin{tabular}{|c|c|c|c|c|c|}
			\hline
			Parameter & $R_{1}$ & $R_{2}$ & $R_{3}$ & $R_{4}$ & $R_{5}$ \\
			\hline 
			$d_{min}^2 - 8-QAM$ & 9.6 & 4.8 & 2.4 & 2.4 & 2.4 \\ 
			
			$d_{min}^2 - 16-QAM$ & 12.8 & 6.4 & 1.6 & 1.6 & 1.6   \\ 
			
			$d_{min}^2 - 32-QAM$ & 15.24 & 3.81 & 0.952 & 0.952 & 0.952 \\ 
			
			$d_{min}^2 - binary$ & 4 & 4 & 4 & 4 & 4 \\ 
			
			\hline
			
		\end{tabular}
		
		\caption \small { Table showing  the minimum distances seen by different receivers for 8-QAM, 16-QAM and 32-QAM in  Example \ref{ex_N_n}.}
		\label{Table_ex_N_n}	
	\end{center}
\end{table}
\end{example}
\section{Conclusion}
\label{sec6}
In this paper we considered noisy index coding over AWGN channel. The problem of finding an optimal index code, for a given index coding problem, is, in general, exponentially hard. However, we have shown that finding the minimum number of binary transmissions required is not required for reducing transmission bandwidth over a noisy channel since, we can  use an index code of any given length as a single QAM point thus saving bandwidth. The mapping scheme by the proposed algorithm and QAM transmission are valid for any general index coding problem. It was further shown that if the receivers have huge amount of side information, it is more advantageous to transmit using a longer index code as it will give a higher coding gain as compared to binary transmission scheme.


\end{document}